\documentclass{article}
\usepackage{arxiv}
\usepackage[utf8]{inputenc}
\usepackage[english]{babel}
\usepackage{algorithm}
\usepackage{algpseudocode}
\usepackage{amsmath}
\usepackage{amssymb}
\usepackage{amsthm}
\usepackage{comment}
\usepackage{bbm}
\usepackage[colorlinks=true, allcolors=blue]{hyperref}
\usepackage{indentfirst}
\usepackage{subcaption}
\usepackage{todonotes}
\usepackage[style=numeric-comp,sorting=none]{biblatex}
\title{An Efficient Algorithm for Unbalanced 1D Transportation}
\author{Gabriel Gouvine}
\date{November 2023}

\theoremstyle{definition}
\newtheorem{theorem}{Theorem}[section]
\newtheorem{definition}{Definition}[section]

\newtheorem{lemma}{Lemma}[section]

\begin{document}

\maketitle

\begin{abstract}
Optimal transport (OT) and unbalanced optimal transport (UOT) are central in many machine learning, statistics and engineering applications.
1D OT is easily solved, with complexity $\mathcal{O}(n\log n)$, but no efficient algorithm was known for 1D UOT.
We present a new approach that leverages the successive shortest path algorithm. By employing a suitable representation, we bundle together multiple steps that do not change the cost of the shortest path.
We prove that our algorithm solves 1D UOT in $\mathcal{O}(n\log n)$, closing the gap.
\end{abstract}

\section{Introduction}

Optimal transport theory was introduced to solve the problem of moving a pile of earth with the minimum amount of work.
Today optimal transport (OT) is used extensively, with applications
in image processing~\cite{papadakis_optimal_2015,bonneel_survey_2023,mori_recognizing_2003,thayananthan_shape_2003,mortensen_sift_2005,ke_pca-sift_2004,lowe_distinctive_2004,bonneel_spot_2019,bai_sliced_2023,pitie_n-dimensional_2005},
electronic circuit design~\cite{brenner_faster_2005, gouvine_coloquinte_2015}
and machine learning~\cite{sejourne_faster_2022,solomon_wasserstein_2014,montavon_wasserstein_2015,arjovsky_wasserstein_2017,liu_wasserstein_2019,tolstikhin_wasserstein_2019,frogner_learning_2015}.

These applications generally consider the optimal transport of mass between two sets of points in $\mathbb{R}^d$ using the euclidean distance.
In practice, it is either solved optimally using the network simplex algorithm, or approximated using various heuristics~\cite{peyre_computational_2020,flamary_pot_2021}.
A special case is $d = 1$, with points on the real line.
This 1D version of OT is often used as a heuristic to approximate solutions to multidimensional problems~\cite{bai_sliced_2023,bonneel_spot_2019,pitie_n-dimensional_2005,gouvine_coloquinte_2015,sejourne_faster_2022}.

When the total supply equals the total demand, the problem is said to be balanced and 1D OT is easy to solve with complexity $\mathcal{O}(n \log n)$.
If the demand is larger than the supply and some of it needs to be ignored, the problem is said to be unbalanced, and no efficient algorithm is known.
This paper presents an algorithm to solve 1D UOT with complexity $\mathcal{O}(n \log n)$, identical to the balanced case.






\section{Prior work}

It is well known that the balanced 1D OT is solvable in time $\mathcal{O}(n \log n)$ by allocating the sources to the sinks in order, with $n$ the total number of sources and sinks~\cite{sejourne_faster_2022}.
However, this fast algorithm does not extend to the UOT, and many recent publications attempt to create fast algorithms for variants of 1D UOT~\cite{sejourne_faster_2022,bai_sliced_2023,bonneel_spot_2019}.

The transportation problem is a special case of minimum-cost flow problem, which is itself a special case of linear programming.
In the general case, a direct application of minimum cost flow algorithms yields a complexity of $\mathcal{O}(n^4 \log n)$~\cite{orlin_faster_1988}. With preprocessing, this can be improved to the best known bound of $\mathcal{O}(n^2 \log n)$ for 1D UOT~\cite{gudmundsson_small_2007}.

The use of the earth-mover distance in machine learning applications has motivated the search for fast approximations to the optimal transport solution.
A common approach is the use of the Sinkhorn distance instead, a regularized version of the transportation distance~\cite{cuturi_sinkhorn_2013}.
The complexity of these approximate algorithms on the 1D UOT is not known~\cite{sejourne_faster_2022}.

After the redaction of this paper, it has been brought to our attention that 1D OT and its variants have been described as transportation problems with a Monge cost matrix~\cite{burkard_perspectives_1996}.
Using this framework, \cite{aggarwal_efficient_1992} claims an algorithm with $\mathcal{O}(n \log n)$ complexity for 1D UOT in 1992.
Given that the algorithm and proof remain vague and that this result is forgotten in the recent literature~\cite{sejourne_faster_2022,bai_sliced_2023,bonneel_spot_2019}, we believe that our paper still fills an important gap in providing a working algorithm for 1D UOT.



\section{Algorithm}

The optimal transport problem is a special case of minimum cost flow problem.
After introducing it formally, we show a simple algorithm derived from the successive shortest path algorithm.
We then show that solutions obtained by this algorithm can be represented as a list of "positions" instead of a network flow.
The steps taken by the algorithms can be seen as simple updates on these positions.
Finally, we describe the algorithm operating on this representation, and prove its complexity.

\subsection{Problem statement}

With $n$ sources and $m$ sinks, the 1D UOT is stated as follow:

\begin{align}
    &\text{minimize } &&\sum_{i = 1}^n \sum_{j = 1}^m \lVert u_i - v_j \rVert x_{ij}\\
    &\text{subject to } && \sum_{j = 1}^m x_{ij} = s_i, & 1 \leq i \leq n\\
    &&& \sum_{i = 1}^n x_{ij} \leq d_j, & 1 \leq j \leq m \label{eq:inequality}\\
    &&& x_{ij} \geq 0, & 1 \leq i \leq n, 1 \leq j \leq m
\end{align}

In this model, $u_i$ (respectively $v_j$) are the positions of the sources (respectively the sinks). We denote $c_{ij} = \lVert u_i - v_j \rVert$ the cost of sending flow from source $i$ to sink $j$.

$s_i$ are the supplies associated with the sources, and $d_j$ are the demands associated with the sinks.

Finally, $x_{ij}$ is the amount of flow sent from source $i$ to sink $j$ in a solution.

  \begin{figure}[H]
	\centering
	
	\includegraphics[width=0.5\textwidth]{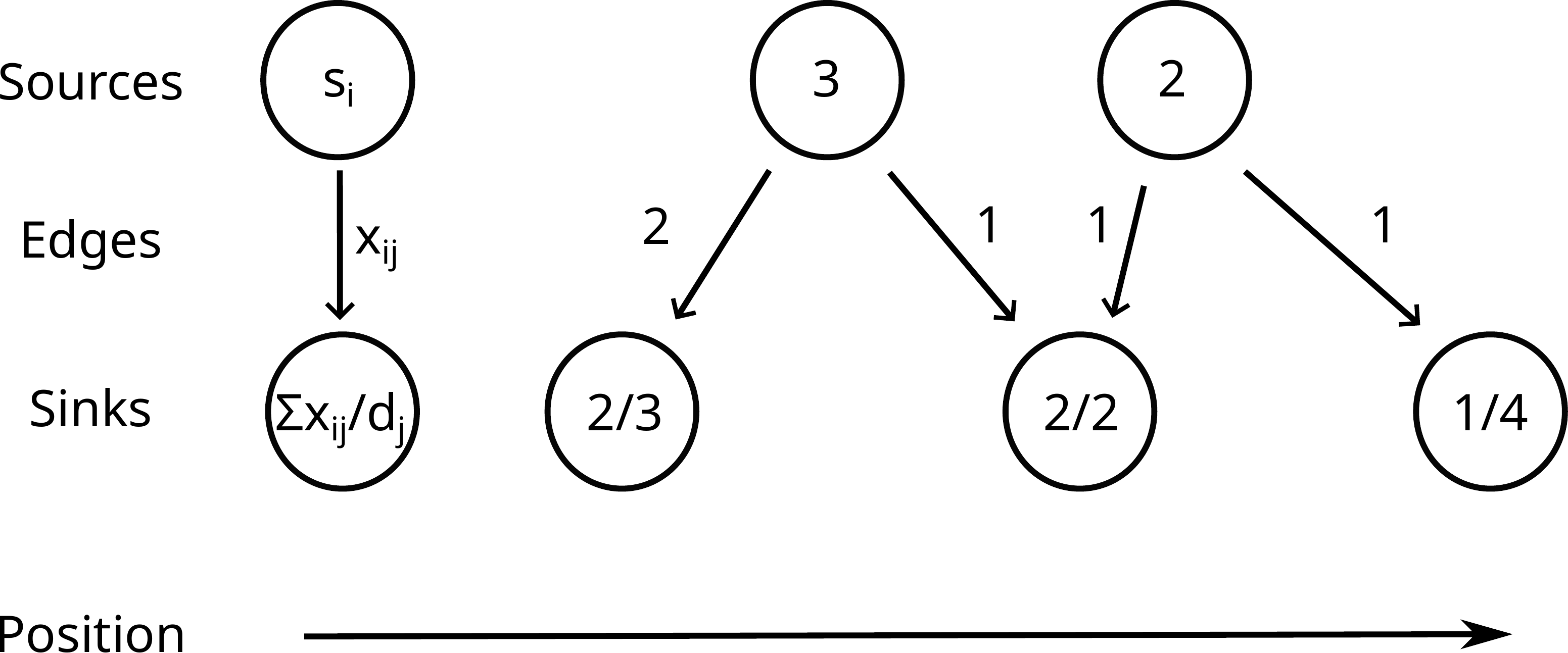}
 \captionsetup{width=.8\linewidth}
 \caption{Representation of a 1D transportation problem and a solution as a flow problem. $u_i$ and $v_j$ values are shown by the position of the node. Edges with zero flow are omitted for clarity.}
 \label{fig:representation}
  \end{figure}

For what follows, we assume demands and capacities to be non-zero, and $u$ and $v$ to be strictly sorted i.e. $u_i < u_{i+1}$ for $1 \leq i < n$ and $v_j < v_{j+1}$ for $1 \leq j < m$.
Both conditions can be enforced on a non-conforming instance by simple pre- and post-processing.

We denote the sums of previous supplies $S_i = \sum_{k=1}^{k\leq i} s_k$ for $0 \leq i \leq n$, and of previous demands $D_j = \sum_{l=1}^{l \leq j} d_l$ for $0 \leq j \leq m$.

Finally, we denote $\delta_{ij} = c_{i j+1} + c_{i+1 j} - c_{i j} - c_{i+1 j+1}$ for $1 \leq i < n$ and $1 \leq j < m$. For network flow algorithms, it is the change in shortest path cost when the flow sent from source $i$ to sink $j+1$ drops to zero, and source $i + 1$ is sent instead.

\subsection{Properties}

\begin{definition}
A solution is \textit{monotonic} when the sources are allocated to the sinks in order, that is $x_{il} > 0 \implies x_{jk} = 0$ for $i < j$ and $k < l$.
\end{definition}

\begin{definition}
A monotonic solution has \textit{no hole} if a source allocated to two sinks or more is given the full capacity of all sinks in between.
That is, $x_{ik} > 0 \text{ and } x_{il} > 0 \implies x_{ij} = d_j$ for $k < j < l$.
\end{definition}

There is always a monotonic optimal solution~\cite{bai_sliced_2023}.
By simple exchange, there is always a monotonic optimal solution with no hole.
The solutions obtained by the algorithm for balanced 1D OT satisfy these properties by construction~\cite{sejourne_faster_2022}.

  \begin{figure}[H]
	\centering
	
\begin{minipage}{.4\textwidth}
	\centering
	\includegraphics[width=0.95\textwidth]{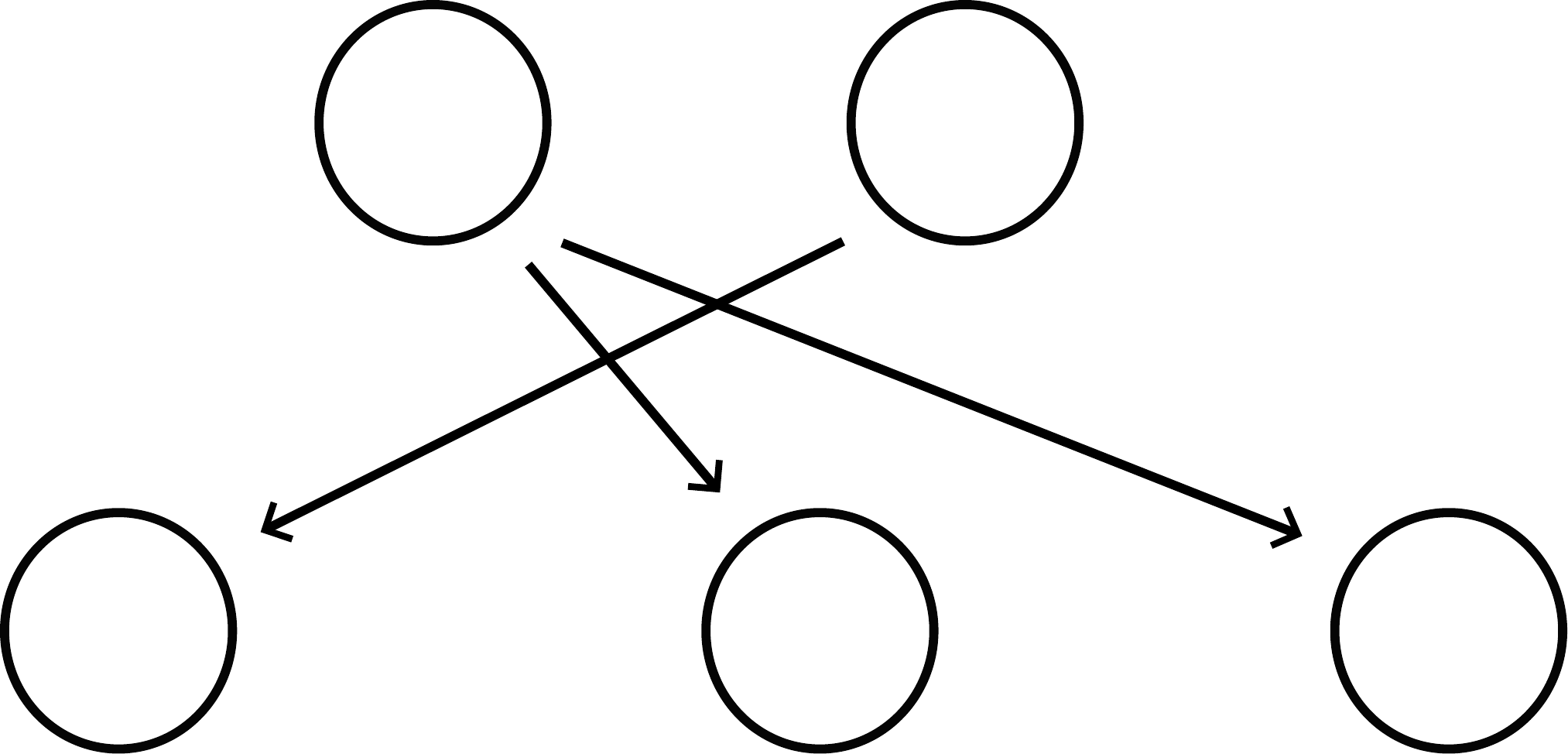}
 \caption{A non-monotonic solution (edges with flow are crossed).}
 \label{fig:non_monotonic}
 \end{minipage}%
\hfil
\begin{minipage}{.4\textwidth}
	\centering
	\includegraphics[width=0.95\textwidth]{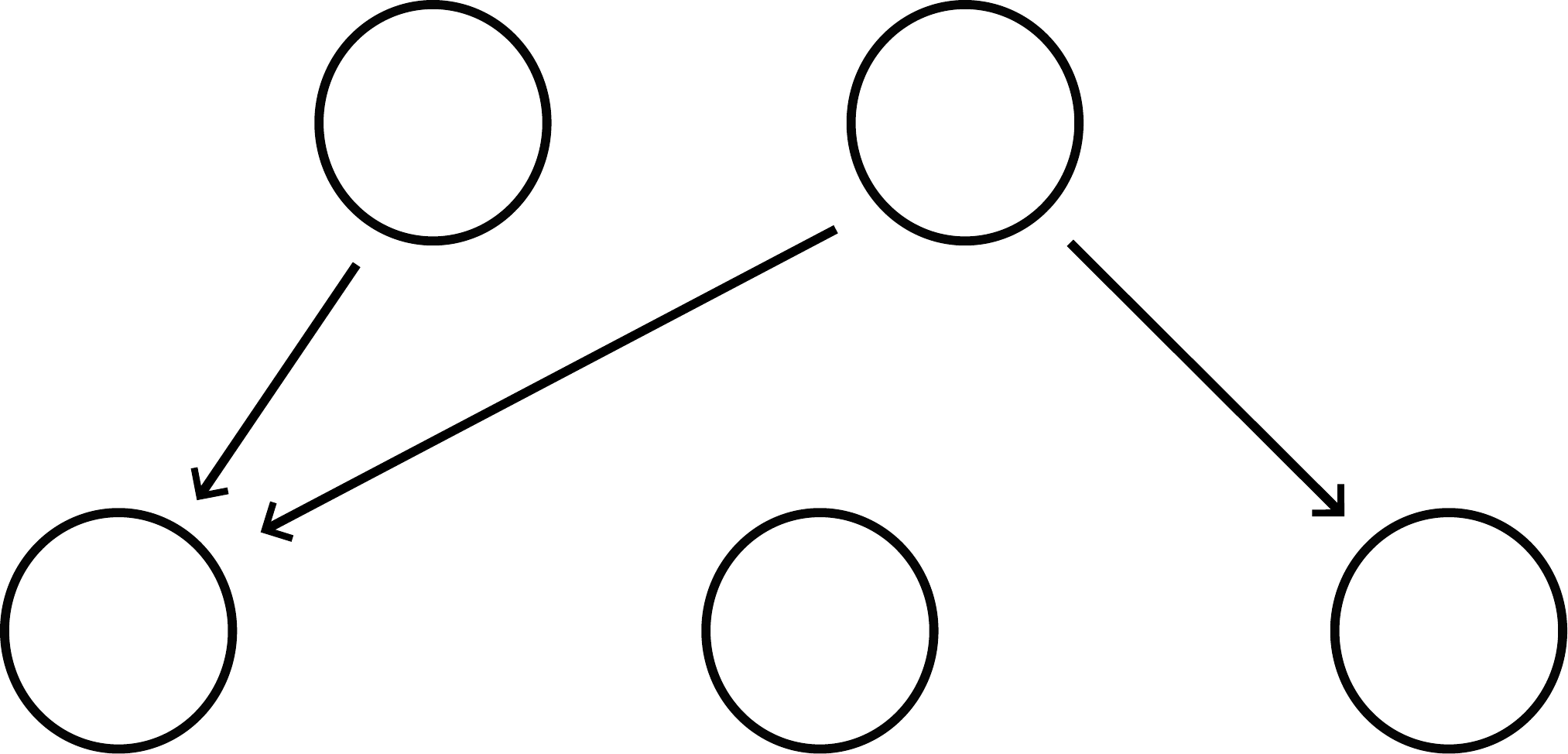}
 \caption{A solution with a hole (the middle sink has unmet demand).}
 \label{fig:hole}
 \end{minipage}
  \end{figure}

\subsection{Successive shortest paths algorithm}

\begin{algorithm}
\caption{1D successive shortest path}
\label{alg:basic_transportation}
\begin{algorithmic}
\For{each source $i$}

\While{$i$ has free supply}
  \State $o \gets \arg\,\min_j c_{ij}$ \text{ the optimal sink for this source}
  \State Find the first sink $j > o$ with free demand
  \If{no such sink exist}
    \State Send supply on the path from the last sink
  \ElsIf{no previous sink has free demand}
    \State Send supply to $j$
  \Else
    \State $c_p \gets$ cost of the path from $j-1$
    \If{$c_p + c_{ij-1} \leq c_{ij}$}
    \State Send supply on the path from $j-1$
    \Else
    \State Send supply to $j$
    \EndIf
  \EndIf
\EndWhile
\EndFor
\end{algorithmic}
\end{algorithm}

We consider a variant of the successive shortest path algorithm (SSP) to solve the unbalanced 1D transportation problem.
This algorithm is a well-known method to minimum-cost-flow problems.

We maintain the optimal solution for the partial problem with $k$ sources and extend it by sending flow from the $k+1$\textsuperscript{th} source.
At each step of the algorithm, it finds the shortest path in the residual graph from the sources to a sink with free demand, then sends as much supply as possible along it.

Let's consider a first version of the algorithm, shown in \autoref{alg:basic_transportation}, where the sources are processed in coordinate order. Given that the partial solution is optimal, each step only has two possible paths to consider:
\begin{itemize}
    \item the single edge to the closest sink, or the first available sink with unmet demand;
    \item the shortest path from the last occupied sink to a sink with free demand, always considering adjacent sinks and the leftmost source allocated to the sink.
\end{itemize}

Two steps of the algorithm are illustrated in \autoref{fig:SSP}.

  \begin{figure}[ht]
	\centering
	
\begin{subfigure}{.4\textwidth}
	\centering
	\includegraphics[width=0.95\textwidth]{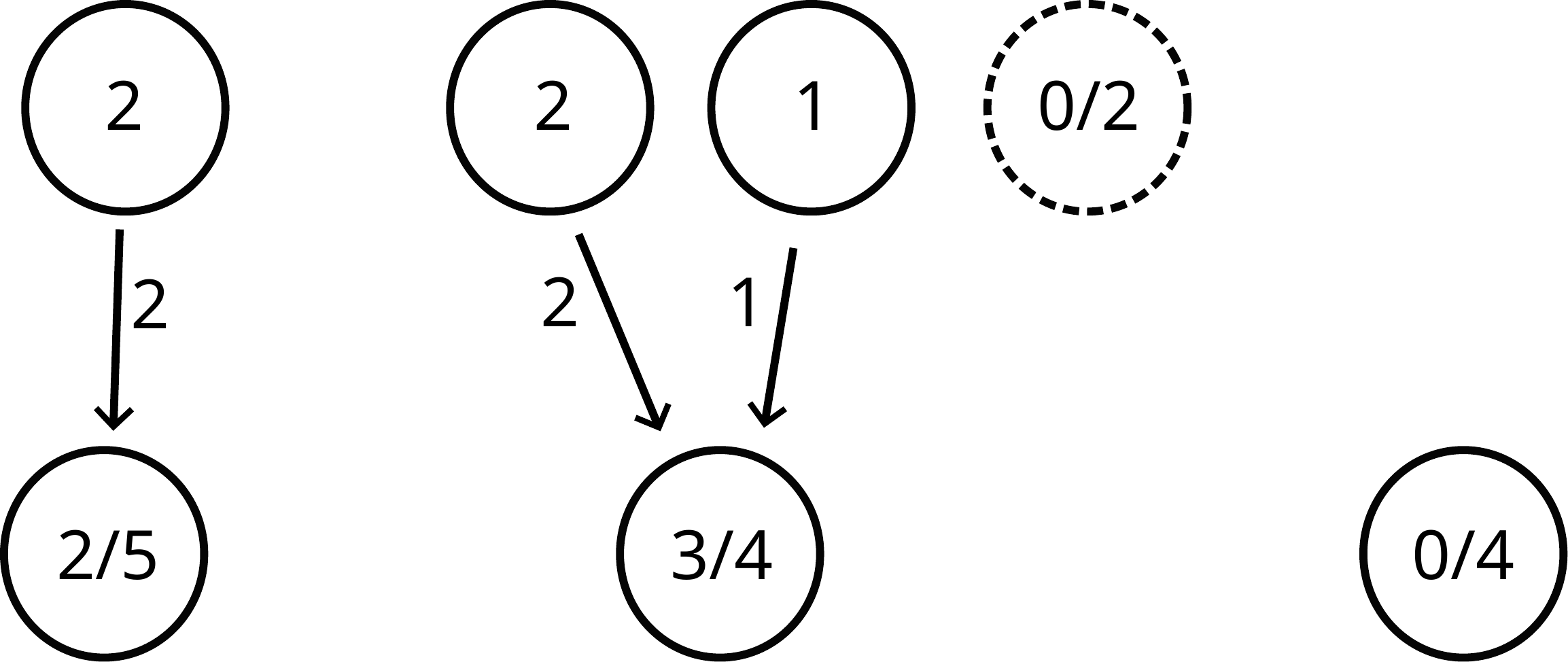}
 \caption{Adding the next source to an existing solution.}
 \end{subfigure}%
\hfil
\begin{subfigure}{.4\textwidth}
	\centering
	\includegraphics[width=0.95\textwidth]{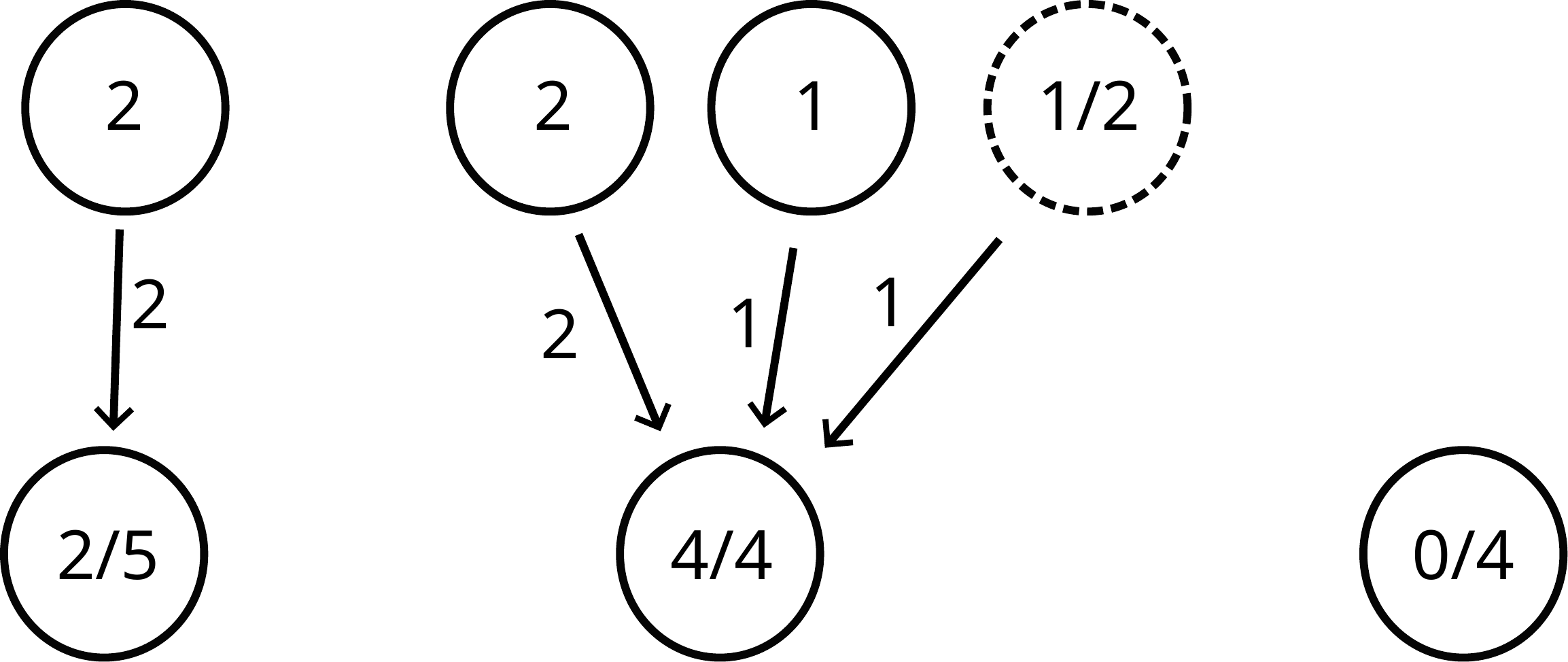}
 \caption{Trivial case where the best sink has unmet demand.}
 \end{subfigure}

\vspace{2em}
 	
\begin{subfigure}{.4\textwidth}
	\centering
	\includegraphics[width=0.95\textwidth]{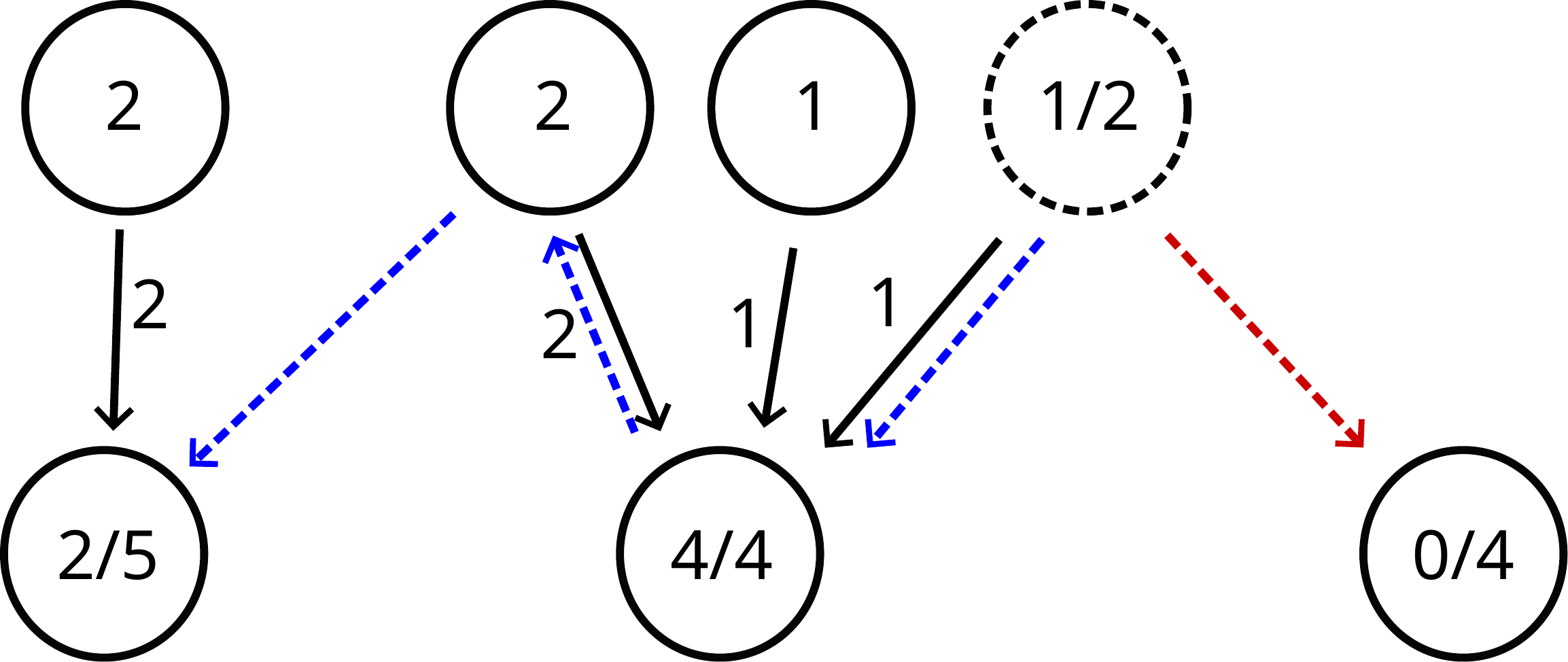}
 \caption{General case: two possible paths to send the remaining supply.}
 \end{subfigure}%
\hfil
\begin{subfigure}{.4\textwidth}
	\centering
	\includegraphics[width=0.95\textwidth]{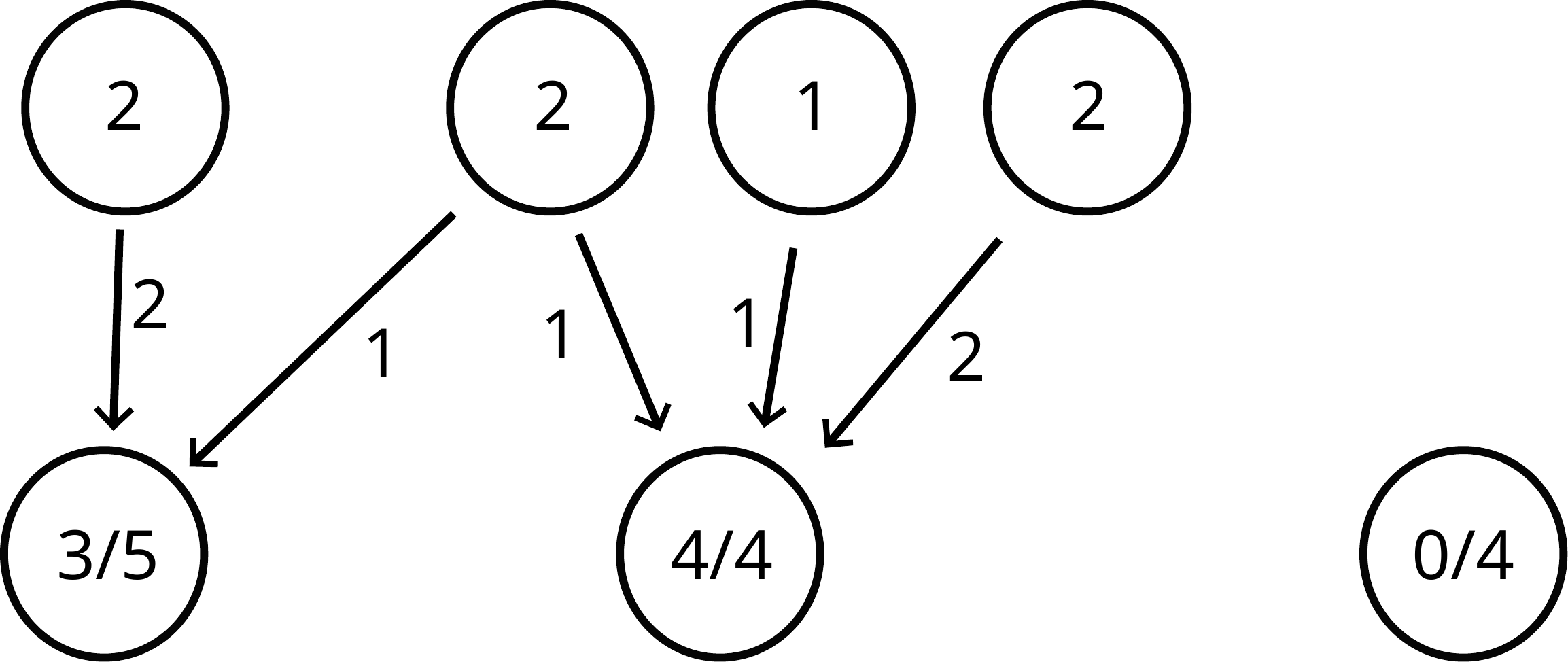}
 \caption{Solution after sending flow on the path to the left.}
 \end{subfigure}
 \caption{Two steps of the SSP algorithm, sending the supply from a new source}
 \label{fig:SSP}
  \end{figure}

\begin{theorem}
    The SSP algorithm in \autoref{alg:basic_transportation} has complexity $\mathcal{O}(nm(n+m))$. 
\end{theorem}

\begin{proof}
Each time flow is sent, either the supply of the source is exhausted, or another source is exhausted and completely leaves one sink for its predecessor. This bounds the number of sending operations by $\mathcal{O}(nm)$ times.

The shortest path has maximum length $\mathcal{O}(m + n)$, so that sending flow takes $\mathcal{O}(m + n)$ operations.
Since the shortest path only involves adjacent sinks and the leftmost source of each sink, we need only consider a topologically sorted subgraph with $\mathcal{O}(m + n)$ edges. The shortest path can then be computed with $\mathcal{O}(m + n)$ complexity.

This yields a total complexity of $\mathcal{O}(nm(n+m))$, which is already better than the direct application of the generic algorithm.
\end{proof}

\begin{theorem}
    The partial solutions maintained by \autoref{alg:basic_transportation} are monotonic and have no hole. 
\end{theorem}

\begin{proof}
Supply from a new source is always sent to the last occupied sink or after, and the leftmost source allocated to a sink is always the first one leaving it, so that the solution remains monotonic.

The new source always sends supply to two adjacent sinks, and new sinks are used only when those are full.
Sending flow on the path never sends flow past a sink with unmet demand, so does not introduce new holes.
\end{proof}

\subsection{Positional encoding for a faster algorithm}

Only a small number of operations change the cost of the shortest path.
We find an implicit representation of the solution, that allows us to only perform work for these operations.

\begin{definition}
A positional encoding is a list $p_i$, $1 \leq i \leq n$ such that $0 \leq \ldots \leq p_i \leq p_{i+1} \leq \ldots \leq D_m - S_n$.
    
\end{definition}

\begin{theorem}
    Positional encodings represent the solutions that are monotonic and have no hole.
    \label{thm:positional-encoding}
\end{theorem}

\begin{proof}
Given an encoding $p$, we create a balanced 1D OT with $n + 1$ additional sources inserted between the original sources, and supplies $p_{i+1} - p_i$.
This yields a problem with $2 n + 1$ sources and $m$ sinks, with supplies $p_0$, ..., $s_i$, $p_{i+1} - p_i$, $s_{i+1}$, ..., $D_m - S_n - p_n$ and demands $d_j$.
The solution to this balanced problem can be found in linear time~\cite{sejourne_faster_2022}, and satifies the monotonicity and no-hole properties.
These properties remain true of the solution restricted to the original sources.

  \begin{figure}[H]
	\centering
	
\begin{subfigure}{.4\textwidth}
	\centering
	\includegraphics[width=0.95\textwidth]{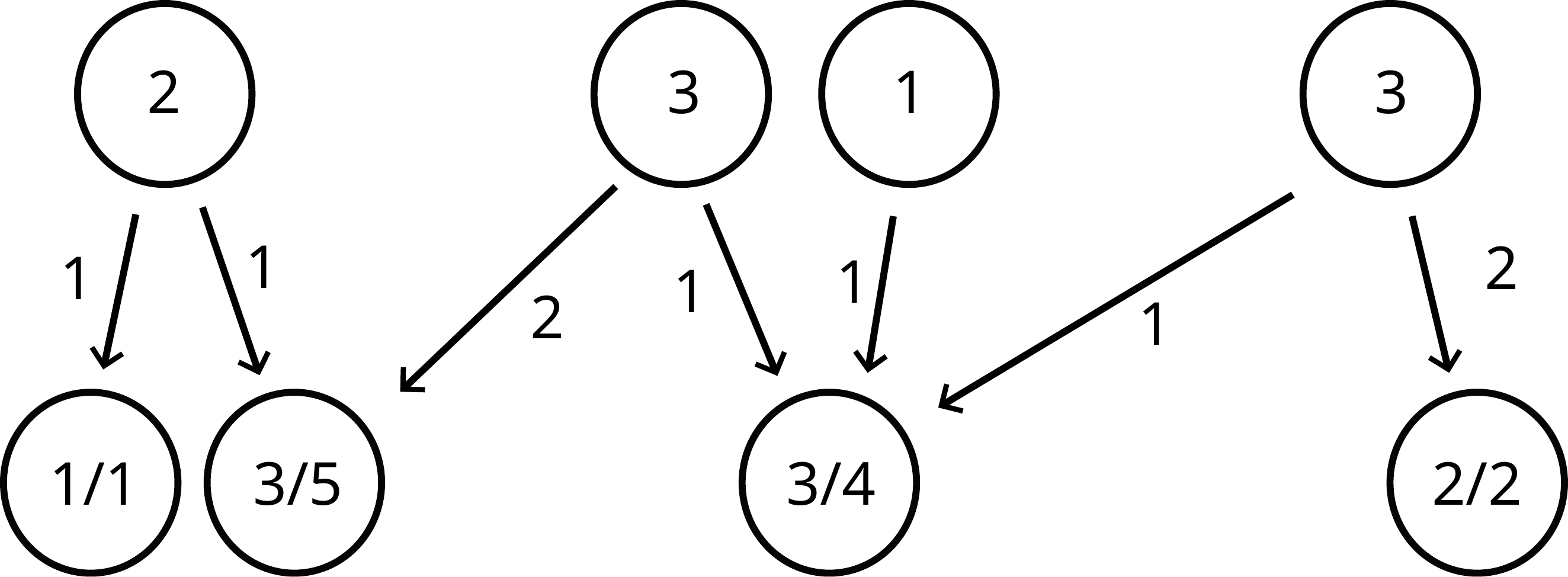}
 \caption{A monotonic solution without holes. A valid positional encoding is $0, 2, 2, 3$}
 \end{subfigure}%
\hfil
\begin{subfigure}{.4\textwidth}
	\centering
	\includegraphics[width=0.95\textwidth]{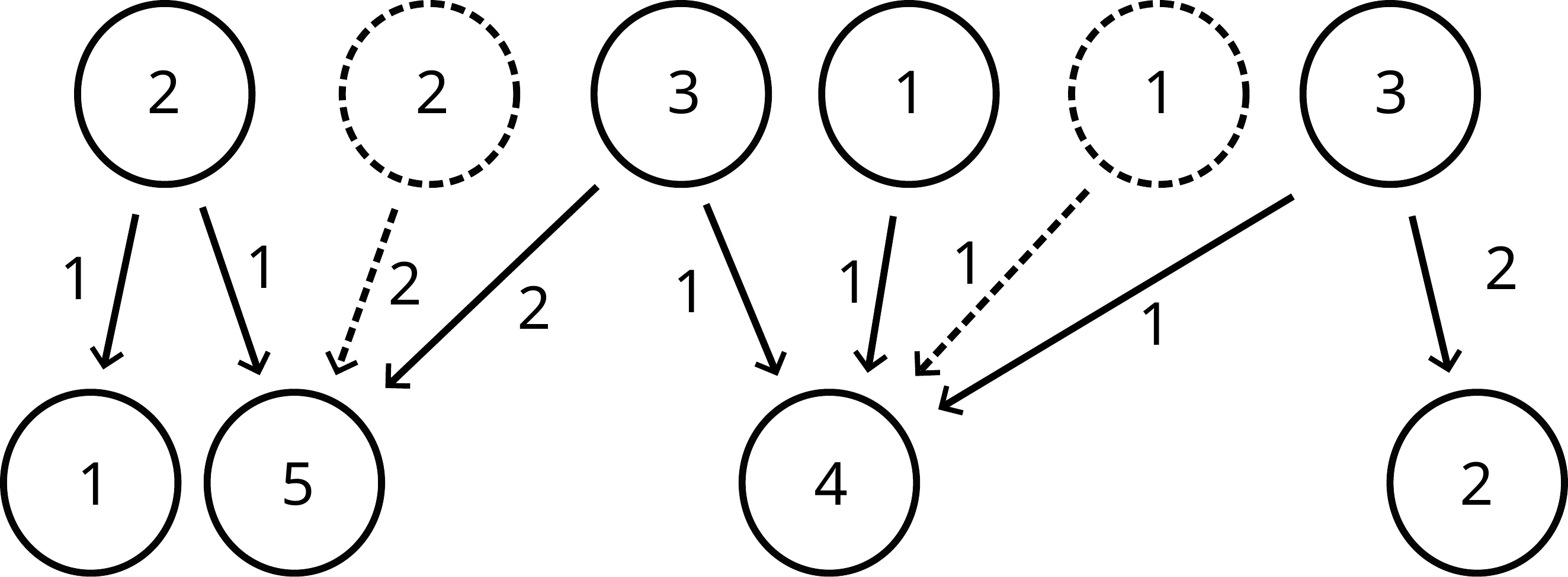}
 \caption{Differences in position correspond to places where more supply could be allocated.}
 \end{subfigure}
 \caption{Visualizing the positional encoding associated with a solution.}
  \end{figure}

Reciprocally, if a solution of the original problem satisfies the monotonicity condition and has no hole, we take $p_i$ such that $p_{i+1} - p_{i}$ represents the unmet demand between two sources.
Note that this representation is not unique, as there are multiple ways to insert the additional supply.
For example, we may use $p_i = D_{j-1} - \sum_{l < j}\sum_{k} x_{kl}$ with
$j = \min_{l, x_{il} \neq 0} l$.
\end{proof}

\begin{theorem}
    When sending flow, all sources on the path share a common position $p^*$.
    A valid encoding after the flow is sent is to decrease $p^*$.
\end{theorem}

\begin{proof}
    Before flow is sent, all sinks on the path except the last one are full: all sources with supply allocated to these sinks must share the same positional encoding.

    After sending flow, the last source on the path may have completely left the sink, so that it may admit other positional encodings.
    If we allocate unmet demand before it in the corresponding 1D OT, there is no gap with the next source and the new value of $p^*$ is a valid encoding.
\end{proof}

\begin{theorem}
    The changes in cost when sending flow happen only for values of $p^*$ of the form $D_{j} - S_{i}$.
\end{theorem}

\begin{proof}
    Changes in costs happen when a edge on the path drops to zero flow ($x_{ij}$ becomes zero i.e. a source leaves a sink) or a new edge is available for use ($x_{ij}$ becomes non-zero i.e. a source enters a sink).
    Since the solution is monotonic, a source $i$ enters a sink $j$ when its position is $D_{j} - S_{i-1}$, and leaves it when its position is $D_{j - 1} - S_{i}$.
\end{proof}

\begin{algorithm}
\caption{Fast 1D successive shortest path}
\label{alg:fast_transportation}
\begin{algorithmic}
\For{each source $i$}
\State $o \gets \arg\,\min_k c_{ik}$ \text{ the best sink for this source}
\State Store changes in costs for source $i$ entering sinks $j$ to $o-1$
\State Store changes in costs between sources $i-1$ and $i$
\State $j \gets \max(j, o)$
\State $p_i \gets \max(D_{j-1} - S_{i-1}, p_{i-1})$ the initial position
\While{$p_i > D_j - S_i$ ($i$ has unallocated supply)}
    \State $\gamma \gets$ cost change at position $p_i$
    \If{$\gamma + c_{ij} < c_{ij+1}$}
        \State Delete the change at $p_i$
        \State Decrease $p_i$ to the first position with a change, or $D_j - S_i$
        \State Increase the cost change at $p_i$ by $\gamma$
    \Else
        \State Increase $j$
        \State Store the change in cost for source $i$ leaving sink $j$
    \EndIf
\EndWhile
\EndFor
\For{each source $i$ from $n-1$ to $1$}
\State $p_{i} \gets \min(p_i, p_{i+1})$
\EndFor
\end{algorithmic}
\end{algorithm}

The fast version of the SSP is shown in \autoref{alg:fast_transportation}.
It keeps the current positional encoding of the solution, on which multiple SSP steps can be performed by simply decreasing the positions.
The positions always decrease and only the last one is used in subsequent iterations, so that it is sufficient to update the last position. The previous elements are updated in linear time once the algorithm has completed.
Our algorithm maintains the cost of the shortest path, and only performs work for positions that correspond to changes in shortest path cost. 

  \begin{figure}[ht]
	\centering
	
\begin{subfigure}{.35\textwidth}
	\centering
	\includegraphics[width=0.95\textwidth]{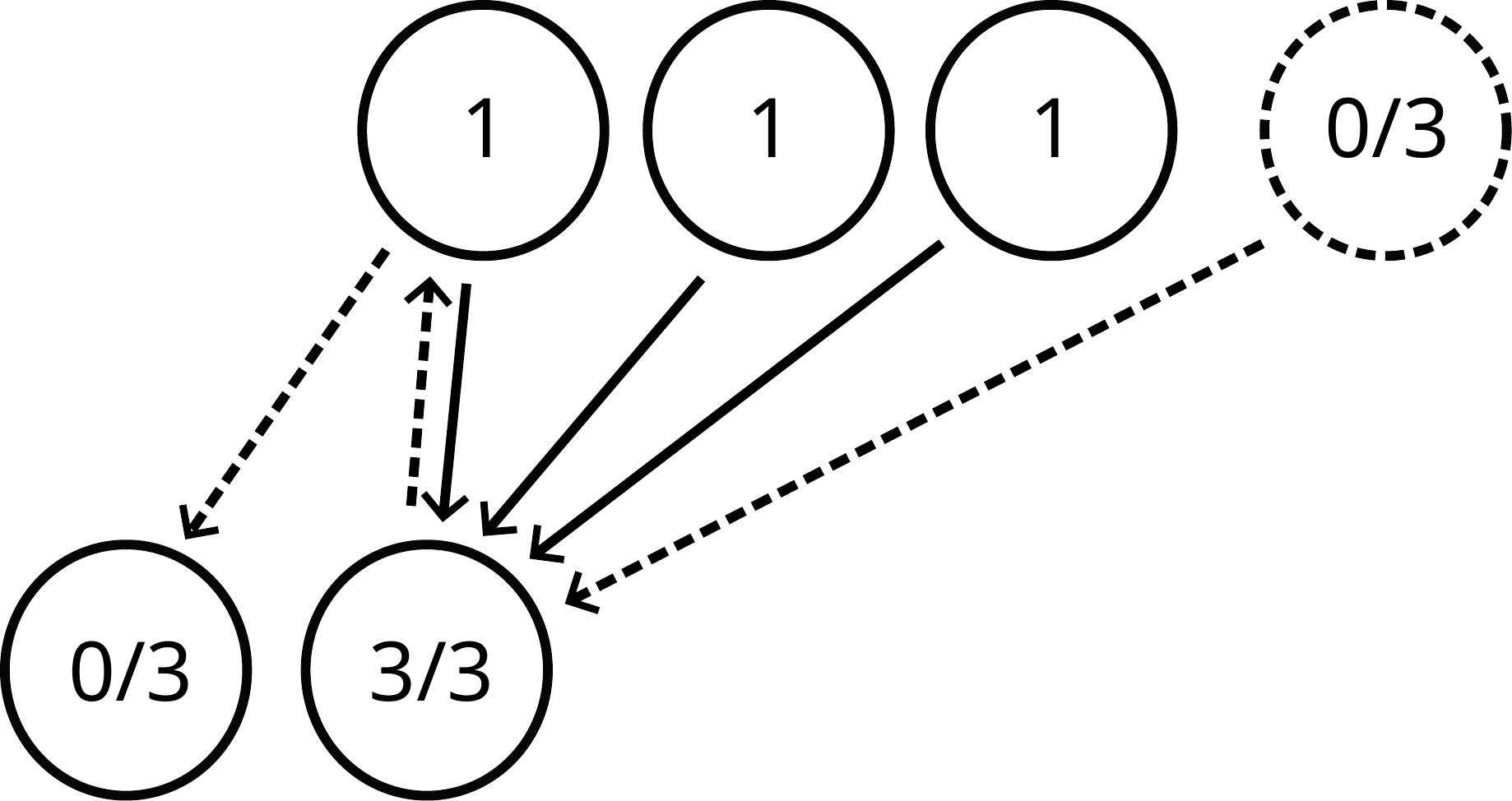}
 \end{subfigure}%
\hfil
\begin{subfigure}{.35\textwidth}
	\centering
	\includegraphics[width=0.95\textwidth]{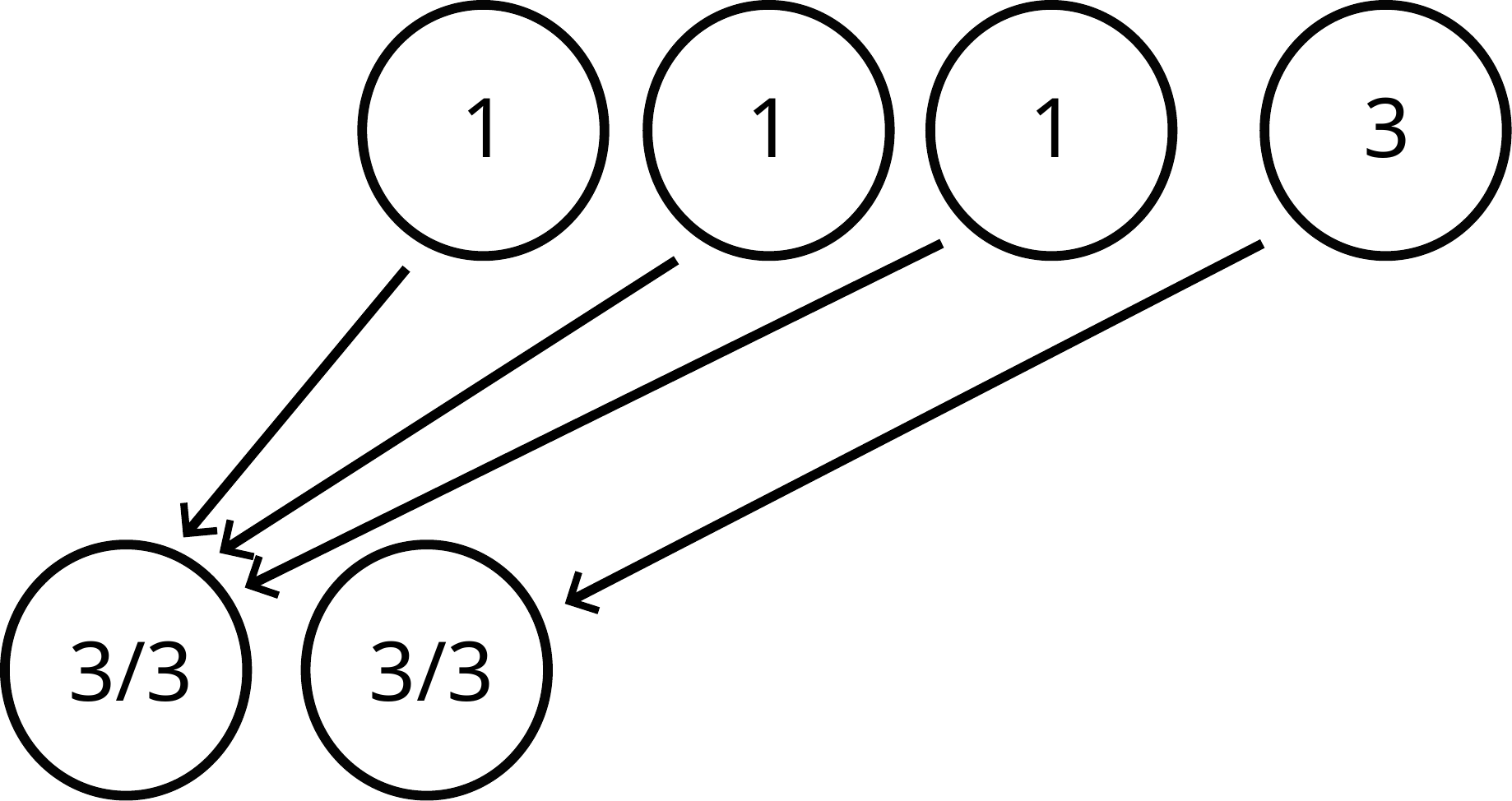}
 \end{subfigure}
 \captionsetup{width=.8\linewidth}
 \caption{On this example, SSP needs to send flow three times (first path shown on the left), but the cost of the path does not change. Our algorithm updates the positional encoding just once, from $3, 3, 3, 3$ to $0, 0, 0, 0$.}
  \end{figure}
  
\subsection{Complexity analysis}
\begin{theorem}
The algorithm has time complexity $\mathcal{O}((n+m) \log(n+m))$.
\end{theorem}

\begin{proof}
The number of iterations and the number of events are $\mathcal{O}(n + m)$:
\begin{itemize}
    \item the changes in cost for source $i$ entering sink $j$ alone  ($c_{ij+1} - c_{ij}$) are added to the list at most once per sink, for a maximum of $m - 1$ events;
    \item the changes in cost between $i-1$ and $i$ ($\delta_{ij}$) are each added to the list at most once, for a maximum of $m + n - 3$ events as shown by \autoref{thm:complexity};
    \item at each iteration of the inner loop, either $j$ increases or an event is removed, for a maximum of $3m + n - 4$ iterations.
\end{itemize}

The operations at each iteration can be performed in $\mathcal{O}(\log (n + m))$:
\begin{itemize}
    \item The events can be stored in a binary tree, with insertion, retrieval in logarithmic time;
    \item Finding the range of non-zero $\delta_{ij}$ and finding the optimal sink for a given source is accomplished with binary search in a sorted list.
\end{itemize}

Finally, the conversion from a positional encoding to a solution can be found in linear time, as described in the proof of \autoref{thm:positional-encoding}.
\end{proof}

\begin{lemma}
The numbers of changes in cost $\delta_{ij}$ that are non-zero is at most $n + m - 3$.
\label{thm:complexity}
\end{lemma}

\begin{proof}
First, note that $\lvert u_i - v_j \rvert - \lvert u_i - v_{j+1} \rvert - \lvert u_{i+1} - v_j \rvert + \lvert u_{i+1} - v_{j+1} \rvert$ is $0$ if $u_{i+1} \leq v_j$ or $u_{i} \geq v_{j+1}$.
It follows that, if $\delta_{ij}$ is nonzero (that is $u_{i+1} > v_{j}$ and $u_{i} < v_{j+1}$), then all $\delta_{kl}$ with either $k > i, l < j$ or $k < i, l > j$ are zero.
In particular, there is no $k, l \neq i, j$ such that $k + l = i + j$ and $\delta_{kl} \neq 0$.

There are $n + m - 3$ possible values for $i + j$, giving us a bound for the number of non-zero elements in $\delta$.
\end{proof}

\begin{theorem}
    The algorithm remains valid when the distance is a convex function other than the absolute value, with complexity $\mathcal{O}(nm \log(n+m))$.
\end{theorem}

\begin{proof}
    In this case, the number of non-zeros in $\delta$ can only be bounded by $nm$.
\end{proof}

\section{Implementation}

We implemented our algorithm in C++. We validate on 10000 randomly generated instances of varying size, using the solutions provided by the Networkx min-cost-flow solver as a baseline~\cite{hagberg_exploring_2008}. The code is available online under an open source license\footnote{\href{https://github.com/Coloquinte/1DTransport}{https://github.com/Coloquinte/1DTransport}}, and used in the electronic placement tool Coloquinte\footnote{\href{https://github.com/Coloquinte/PlaceRoute}{https://github.com/Coloquinte/PlaceRoute}}.

\section{Conclusion}

In this paper, we presented an algorithm to solve the 1D unbalanced optimal transport problem in $\mathcal{O}(n\log n)$ time, closing the gap with the balanced case and improving the best known algorithm by a linear factor.

Our method takes advantage of the successive shortest path algorithm on the corresponding network flow problem.
We gave an equivalent representation of the solutions generated by the algorithm, and show that the steps taken by the algorithm are simple updates on this representation.
This led to an algorithm where most of the updates can be applied as a single step.
We then proved its complexity.

Our algorithm exploits the 1D structure of the problem, and does not generalize to higher dimensions. Future work could focus on improving the current algorithms for two dimensions and more.

\clearpage
\printbibliography

\end{document}